\documentclass{acmsmall-ec13}
\usepackage{fancyhdr}
\usepackage[numbers]{natbib}
\usepackage{graphics}

\usepackage{boxedminipage}
\usepackage{a4}
\usepackage{latexsym}
\usepackage{amssymb}
\usepackage{amsmath}
\usepackage{color}
\usepackage{verbatim}
\usepackage{euscript}
\usepackage{float}
\usepackage{mathrsfs}
\usepackage{euscript}
\usepackage{graphicx}
\usepackage{epstopdf}
\usepackage{bm}
\usepackage{psfrag}
\usepackage{url}
\usepackage{textcomp}

\renewcommand{\epsilon}{\varepsilon}
\renewcommand{\rho}{\varrho}
\renewcommand{\phi}{\varphi}

\renewcommand{\tilde}{\widetilde}

\renewcommand{\bar}[1]{\mkern
  1.5mu\overline{\mkern-1.5mu#1\mkern-1.5mu}\mkern 1.5mu}
\renewcommand{\leq}{\leqslant}
\renewcommand{\geq}{\geqslant}

\newcommand{\E}{{\mathbb E}}

\newcommand{\I}{{\mathbb I}}

\newcommand{\clS}{{\mathcal S}}

\newcommand{\nn}{\nonumber}

\renewcommand{\t}{\theta}
\newcommand{\bt}{\bm{\theta}}
\renewcommand{\a}{\alpha}
\newcommand{\G}{\Gamma}

\newcommand{\bw}{\bm{w}}

\newcommand{\bart}{\bar{\t}}
\newcommand{\beq}{\begin{equation}}
\newcommand{\eeq}{\end{equation}}

\newcommand{\bti}{\bt_{-i}}

\newcommand{\arule}{x}

\numberwithin{equation}{section}

\begin{document}
\lhead{} 
\rhead{} 
\lfoot{} 
\cfoot{} 
\rfoot{} 
\chead{\bfseries Proceedings Article}  
\title{Ranking and Tradeoffs in Sponsored Search Auctions}
\author{BEN ROBERTS \affil{Microsoft Research Cambridge}
DINAN GUNAWARDENA \affil{Microsoft Research Cambridge}
IAN A. KASH \affil{Microsoft Research Cambridge}
PETER KEY\affil{Microsoft Research Cambridge}}

\begin{abstract}
In a sponsored search auction, decisions about how to rank ads impose
tradeoffs between objectives such as revenue and welfare. In this
paper, we examine how these tradeoffs should be made.  We begin by
arguing that the most natural solution concept to evaluate these
tradeoffs is the lowest symmetric Nash equilibrium
(SNE). As part of this argument, we generalise the well known connection
between the lowest SNE and the VCG outcome.
We then propose a new ranking algorithm, loosely based on the revenue-optimal
auction, that uses a reserve price to order the ads (not just to
filter them) and give conditions under which it raises
more revenue than simply applying that reserve price. 
Finally, we conduct extensive simulations examining the tradeoffs enabled by
different ranking algorithms and show that our proposed algorithm enables
superior operating points by a variety of metrics.
%
\end{abstract}



\keywords{Sponsored search, keyword auctions, generalized second price, reserve price}


\begin{bottomstuff}
\end{bottomstuff}

\maketitle
\thispagestyle{fancy}

\section{Introduction}

In a sponsored search auction, the mechanism designer is given a search query from a user and
bids from advertisers and faces the problem of
deciding which ads to show in response and where to place them. His
choices govern the balance achieved between the interests of stakeholders: the user wants meaningful content, the
publisher wants revenue, and the advertisers want engagement.
In current sponsored search
auctions run by Google or Bing, ads' placement is determined by a ranking algorithm.  This ranking then determines prices through the payment rule of the generalised
second price (GSP) auction.  Payment is made when a user clicks on an ad (pay-per-click).  Hence an ad's ranking affects the user, both through a direct position effect and an indirect externality on non-sponsored or algorithmic content, has a consequent affect on the advertiser through the probability of engaging with or clicking on the ad, and affects the price paid by the advertiser and revenue generated to the publisher through the pricing mechanism.  Higher positions typically receive more attention from users and more clicks.   Different ranking algorithms enable different tradeoffs to be made among these stakeholders.  In this paper, we examine the tradeoffs enabled by different ranking algorithms and use a combination of theoretical analysis and empirical study to gain insights into new ranking methods.


 Historically, Yahoo! initially ranked advertisers by bid  $\{b_i\}$, then Google adopted a ranking based on value per impression, or expected revenue, namely ranking by $\{b_iw_i\}$, where the
\textit{relevance} $w_i$ is an estimate of the likelihood that
advertiser $i$'s ad will be clicked.   Since then, there has been a
widespread adoption of this rank-by-expected revenue approach,  subject to various enhancements  such as setting reserves prices.

The framework for much of our paper is a single instance of an auction
where advertisers bid for slots, which differ in their ``quality''
(such as click-through rate).  This is an abstraction of live
ad-auctions, where advertisers bid on keywords as part of a campaign,
and keywords are matched to queries. We sidestep the matching
issues and assume that the bids represent bids for the actual query,
which is indeed the case for exact match bids on specific keywords.   In
the analysis phase, we assume many of the auction parameters are
fixed, if unknown: for example quality scores, click-through rates and
the number of bidders.    Later in the paper we show simulation
results using live data, where these assumptions certainly do not
hold.

In order to examine the necessary tradeoffs, we need a solution concept that describes the outcome we expect.  The standard analysis of GSP auctions~\cite{aggarwal06,edelman,varian} looks at complete information Nash equilibria,  and in particular their refinement to {\em symmetric} or {\em locally envy-free} Nash equilibria (SNE).  However, the standard approach to analyzing the revenue of auction designs, based on Myerson's work on optimal auctions~\cite{myerson}, examines performance in Bayes-Nash equilibria.

Fortunately, it turns out that we do not need to choose between these two solution concepts.  In a striking result, several groups of authors independently showed that, when ads are ranked by their bid multiplied by their click probability, the ``lowest'' SNE corresponds to the the VCG outcome~\cite{aggarwal06,edelman,varian}.  Aggarwal et al. showed that this continues to hold when these ``rank scores'' are multiplied by individualised weights, except that rather than VCG the results correspond to the outcome of a truthful mechanism they call the ``laddered auction.''  As this is a single parameter domain, this mechanism is a special case of Myerson's general technique for transforming a monotone allocation rule into a truthful mechanism~\cite{myerson}.

Our first result futher generalises the connection between the Myerson
outcome and the lowest SNE of the GSP auction. We show that the it
holds for a broad class of ranking algorithms, which is large enough
to to include almost all of those that have been previously considered
in the literature.  We discuss how an argument due to Edelman and
Schwartz~\shortcite{edelman2} provides additional justification for
our adoption of the lowest SNE as our solution concept; it also
enables us to bound the performance of ranking algorithms outside our
class, where SNE may not exist. A notable algorithm of this type is the practically
important case of simply imposing a minimum bid on the
rank-by-expected revenue system.

Our second contribution is a proposal of a new ranking algorithm,
inspired by features of the revenue-optimal auction, with provably
good properties.  Rather than using a reserve price simply as a
minimum bid, we  incorporate it directly into the ranking algorithm
such that ads with bids near the reserve price receive low rank scores
relative to ads with high bids but lower click probabilities. That is,
we  recommend the class of ranking algorithms $\{(b_i-r)w_i\}$, where $r$
is the reserve price.  This change from $\{b_iw_i\}$ achieves a similar
``squashing'' effect to introducing an exponent $\a<1$ and  ranking by $\{b_iw_i^\a\}$, proposed by Lahaie and Pennock~\shortcite{lahaie}.  Hence  both squashing and
setting a reserve are achieved through a single parameter, as
opposed to the effects being decoupled into a squashing exponent
and a reserve.

We prove that, for sufficiently small reserve prices, incorporating the
reserve price into the ranking algorithm raises more revenue than simply
using that same reserve price purely as a filter (i.e.\ the additional revenue is due the change in ordering, not merely the introduction of a reserve price).  The meaning of sufficiently small depends on the distribution of advertiser valuations, but for a number of natural distributions (e.g. uniform and exponential) it encompasses all choices of reserve price that do not exceed the revenue-optimal reserve.
While our ranking is inspired purely by the revenue-optimal auction, this theorem provides some insight into why it may enable favourable tradeoffs: to raise a given amount of revenue, a lower (and thus less distortionary) reserve price can be used.

Our third contribution is an extensive consideration of the tradeoffs
enabled by different ranking algorithms through simulations.  As the
auctioneer cares about both short-term revenue and the long-run health
of the search platform, the relevant operating points likely
do not include the welfare-optimal and revenue-optimal designs.  This
is in contrast to previous work, which tends to focus on performance
only at these extreme points.  We use two types of simulation: one
simulating a reactive environment where users react to changes in
auction design by forming revised equilibria; the other using real
auction logs which assumes advertisers do not react to
changes, but which captures all the vagaries of real auctions. These include budget
constraints, a changing set of advertisers, and  stochastic quality factors.    These simulations 
show that our proposed ranking algorithm enables favourable tradeoffs
using a variety of metrics.  Our results also show that our proposed
ranking
has several nice properties from an optimisation perspective.

\section{Preliminaries}
\label{s-model}

We adopt a standard (Bayesian) model of a GSP auction:
\begin{itemize}
\item There are $n$ advertisers (bidders) and $m$ slots.
\item If bidder $i$'s ad is displayed in slot $k$, its
 \textit{click-through rate}
  (CTR) is $w_is_k$. $s_k$ is a slot effect, while
  $w_i$ is an ad effect and can be interpreted as the
  \textit{relevance} of bidder $i$'s ad. The slots are strictly
  heterogeneous, with effects $s_1> s_2>\cdots$.
\item Advertiser $i$ has value $\t_i$ for a click. Values are
  i.i.d.\ with cdf $F(\t_i)$ and pdf $f(\t_i)$.
\item Advertisers are assigned to slots by a ranking algorithm.
  This can be represented by a \emph{ranking function} $y(b,w)\geq 0$.  The advertisers
  are sorted by $y(b_i,w_i)$ with the highest score receiving the first
  slot.  Advertisers with $y(b_i,w_i) = 0$ are excluded (e.g.\ if they
  are below some reserve).
We restrict $y$ to be a monotone function
  with respect to $b$, but not necessarily $w$.
\item Advertisers pay the {\em generalised second price} for their slot.  Assuming for simplicity
  that advertisers are ordered such that advertiser $i$ is in slot $i$, advertiser $i$'s payment is
  the minimum bid needed to keep his slot
\beq\label{e-GSPpi}
p_i^y(b_{i+1},w_{i+1};w_i) = \inf\bigl\{b\,:\,y(b,w_i)>y(b_{i+1},w_{i+1})\bigr\}\,.
\eeq
\item We ignore the possibility of non-trivial ties (i.e.\
  $y(b_i,w_i)=y(b_j,w_j)>0$) as they complicate analysis, and
it is not clear how ties should be resolved in a GSP auction. Our
analysis focuses on performance \emph{in expectation}, and so we justify this
oversight by noting that for all our considered ranking functions,
non-trivial ties occur with probability zero and have no bearing on
any expected quantity.
\item In our analysis, it will be helpful to refer to the slot effect $s_k$
  assigned to advertiser~$i$ as his \textit{allocation}
  $\arule_i$. In some instances it is appropriate to consider this a
  function of the realisation of advertiser types: $\arule_i(\bt,\bw)$. If a
  ranking function $y(b,w)$ is used to assign the slots, then it is
  appropriate to consider an advertiser's allocation as a function of
  the bid and relevance vectors, which we write as
  $\arule_i^y(\bm{b},\bw)$.
\item We use $\arule$ (or $\arule^y$) to denote an \emph{allocation rule}, which comprises
  the set of allocation functions $\{\arule_i\}$ (or $\{\arule_i^y\}$).
\end{itemize}
The only new feature of this model is the use of a general ranking algorithm based on $y$ rather
than assuming a particular instantiation. 

Much of our theoretical work utilises \emph{virtual values}, a common
concept in economic theory. Under general conditions, an advertiser's
virtual value may depend on both his true value and his
relevance. However in the interests
 of an easier analysis we assume
independence between these two variables, which defines an
advertiser's virtual value also to be independent of his
relevance:
\beq\label{e-phi}
\phi(\t_i) = \t_i-\frac{1-F(\t_i)}{f(\t_i)}\,.
\eeq
We further assume that the virtual value function is differentiable, and
that the hazard rate $f(\t_i)/(1-F(\t_i))$ is
non-decreasing, conditions which hold for numerous common distributions. Again,
these assumptions are to accommodate an easier analysis. The value at
which an advertiser's virtual value becomes zero is $\bart$. That is,
$\phi(\bart)=0$. Given
that $\t_i$ has distributional support at zero, our assumptions imply
that $\bart$ exists and is unique.


\section{Proposed Ranking Algorithm}

In order to compare ranking algorithms, we must make some assumption about
bidder behaviour. A useful starting point is to assume a Bayes--Nash
equilibrium (BNE) in which each advertiser submits a bid maximising his own
benefit in expectation over the others' types and bids, and his own relevance.
%
An advantage of working in the Bayesian setting is that we can use
Myerson's~\shortcite{myerson} theory to quickly calculate expected
revenue $R$. In any BNE,
\beq
R=R(\arule)=\E\biggl[\sum_{i=1}^n\phi(\t_i)w_i\arule_i(\bt,\bw)\biggr]\,,\label{e-R}
\eeq
where we write $R(\arule )$ instead of  $R(\arule ;f,\bw)$ to
emphasise the dependence on the allocation rule $\arule$.

Hence using \eqref{e-R}, one can simply characterise the
revenue-optimal auction. That is, it ranks advertisers by 
$\phi(\t_i)w_i$, excluding any advertiser with negative virtual value
(i.e.\ the auction has a reserve price of $r=\bart$ where $\phi(\bart)=0$).
However, actually implementing this auction is unlikely
to be feasible in practice.  In particular, this simple
form relies on our assumptions that bidders are symmetric, and that relevance and
value are independent.  Otherwise, virtual values
(and hence the ranking and reserve price) depend on the identity and
relevance of the bidder, which makes practical auction design difficult.
Even if we could implement such a  revenue-optimal auction, other considerations such as
advertiser and user satisfaction would make doing so undesirable.

Instead, we note two qualitative features of the revenue-optimal auction.  First,
it uses a reserve price.  Second, bidders with values barely above the reserve
price are very low in the rankings.  This inspires the new ranking algorithm we evaluate in
Sections~\ref{s-revenue} and~\ref{s-example}
which ranks ads by
$\{(b_i-r)w_i\}$, which is perhaps the simplest ranking with these two features.
Note that under our proposal the price paid for slot $i$ is   $b_{i+1}( {w_{i+1}}/{w_i}) + r( 1- {w_{i+1}}/{w_i} )$,
which follows from  \eqref{e-GSPpi}, assuming  advertiser $i$ is allocated slot $i$.

\section{The Lowest Symmetric Nash Equilibrium}

We now identify one specific Nash equilibrium that we shall use to compare ranking algorithms, which has appealing features.  We  also draw the connection to  BNE.

Because of the  difficulties involved in a full  Bayes-Nash analysis for the GSP auction, a  commonly used
alternative is to assume a \textit{symmetric Nash
  equilibrium} (SNE), an ex-post equilibrium concept proposed
independently by Varian~\shortcite{varian}
and Edelman et al.~\shortcite{edelman}, (who used the term \textit{locally
  envy-free equilibrium}). A SNE requires the following inequalities
to be satisfied:
\beq\label{e-SNEineq}
\bigl(\t_i-p_i^y(b_{i+1},w_{i+1};w_i)\bigr)x_i\geq
\bigl(\t_i-p_i^y(b_{j+1},w_{j+1};w_i)\bigr)x_j\quad\text{for all}\,\,i,j\,,
\eeq
where $p_i^y$ is the GSP
payment~\eqref{e-GSPpi}. Note that the SNE inequalities~\eqref{e-SNEineq} are
stronger than those that define an ex-post Nash equilibrium, which for $j<i$
would replace the subscripts $j+1$ with $j$ in the right hand side of~\eqref{e-SNEineq}.  Hence the set
of SNE is a subset of the set of ex-post equilibria.

The striking result that both Varian~\shortcite{varian} and Edelman et
al.~\shortcite{edelman} realised is that under the ranking algorithm
$\{b_iw_i\}$ any SNE yields an efficient outcome, and furthermore
there exists a SNE --- known as the
\textit{lowest} or \textit{bidder-optimal} SNE --- in which
advertisers' positions and payments are identical to those imposed by the
VCG mechanism.
Aggarwal et al.~\shortcite{aggarwal06} showed a more general connection between the ranking algorithm $\{b_iw_ic_i\}$ (where $c_i$ is a positive constant) and the corresponding ``laddered auction'', a family of truthful mechanisms.
This result is
important for a number of reasons:
\begin{itemize}
\item It provides a focal outcome from the
  space of possible SNE. 
\item It creates a link between SNE behaviour and the Bayesian
  setting.
\item It provides a natural lower bound on revenue, as every other SNE has higher revenue.
\end{itemize}


We show that this result is much more general.  In particular, for any ranking
function of the form
\beq\label{e-y}
y(b,w)=\bigl(g(w)b-h(w)\bigr)^+\,,
\eeq
SNE always exist ($g$ and $h$ are arbitrary non-negative functions).
 Since $y$ does not does not necessarily rank the best ad highest, the
outcome is,  in general, no longer efficient.  However,
it does respect $y$, in the sense that the ranking in all SNE is the same ranking
that would be used if bidders reported their true values.  Finally, the lowest
SNE still has a very special structure.  Recall that
$\arule^y(\bm{b},\bw)=\{x_i^y(\bm{b},\bw)\}$ are the allocations 
that result from  bids $\bm{b}$ and ranking function $y$.  By~\eqref{e-y}, 
$\arule^y(\bm{b},\bw)$ is a monotone allocation rule.  Therefore there are unique
payments that make $\arule^y$ an ex-post direct revelation mechanism.  The lowest
SNE implements this mechanism in the exact same way that the standard ranking
implements VCG.  In particular, since ex-post direct revelation
mechanisms are also BNE, this
allows us to give a concise characterisation of the revenue in the lowest SNE.

\begin{theorem}
\label{thm:lowest}
Consider a GSP auction subject to a
ranking algorithm $y(b,w)$ within the class \eqref{e-y}.%
\footnote{For simplicity, we assume that all bidders are ranked by the same algorithm $y$.  However, our result still holds if each is ranked using an individualised algorithm $y_i$ from the class \eqref{e-y}.  This enables our result to apply to settings where, for example, the mechanism design incorporates other factors into the rank score.}
For any realisation $(\bt,\bw)$, there exists a non-empty set of SNE
and all SNE order bidders by $y(\t_i,w_i)$.
Furthermore, the lowest (revenue) SNE,
defined by 
\beq
y(b_i,w_i)x_{i-1}=\sum_{j\geq
  i}y(\t_j,w_j)(x_{j-1}-x_j)\,,\label{e-lowSNE}
\eeq
generates expected revenue 
\beq
R(\arule^y) =\E\biggl[\sum_{i=1}^n\phi(\t_i)w_i\arule^y_i(\bt,\bw)\biggr]\,.\label{e-R2}
\eeq
\end{theorem}
\begin{proof}
From the GSP payment rule \eqref{e-GSPpi}, the price-per-click charged to
bidder $i$ is
\[                                                                                     
p_i^y(b_{i+1},w_{i+1};w_i)=\frac{y(b_{i+1},w_{i+1})+h(w_i)}{g(w_i)}\,\,.
\]
The SNE inequalities~\eqref{e-SNEineq} are then
\[
\Bigl(\t_i-\frac{y(b_{i+1},w_{i+1})+h(w_i)}{g(w_i)}\Bigr)x_i\geq \Bigl(\t_i-\frac{y(b_{j+1},w_{j+1})+h(w_i)}{g(w_i)}\Bigr)x_j\,,
\]
which is equivalent to
\beq\label{e-SNEineq2}
\bigl(y(\t_i,w_i)-y(b_{i+1},w_{i+1})\bigr)x_i\geq \bigl(y(\t_i,w_i)-y(b_{j+1},w_{j+1})\bigr)x_j\,.
\eeq
Varian's~\shortcite{varian} analysis can be directly reapplied to this
generalisation, leading to the conclusion that there exists a
non-empty set of SNE, and furthermore all SNE use the same allocation
rule, ordering bidders by $y(\t_i,w_i)$.

In the lowest SNE \eqref{e-lowSNE}, advertiser $i$'s payment $p_i$ satisfies
\begin{align*}
y(p_i,w_i)=y(b_{i+1},w_{i+1})&=\frac{1}{x_i}\sum_{j\geq
  i+1}y(\t_j,w_j)(x_{j-1}-x_j)\\
&=y(\t_i,w_i)-\frac{1}{x_i}\sum_{j\geq
  i}x_j\bigl(y(\t_j,w_i)-y(\t_{j+1},w_{j+1})\bigr)\\
&=y(\t_i,w_i)-\frac{1}{x_i}\int_0^{\t_i}\!\!\!\arule_i^y(t,\bti,\bw)\,dy(t,w_i)\,.
\end{align*}
As $dy(t,w_i)=g(w_i)\,dt$,
\[
p_i=\t_i-\frac{1}{x_i}\int_0^{\t_i}\!\!\!\arule_i^y(t,\bti,\bw)\,dt\,,
\]
which precisely describes the payment functions imposed by the ex-post direct
revelation mechanism for the allocation rule
$\arule^y(\bt,\bw)$. Thus, the lowest SNE is also a BNE, and therefore
generates expected revenue \eqref{e-R2}.
\qed
\end{proof}

This generalisation of the lowest SNE to the class of
rankings~\eqref{e-y} includes ranking by bid $\{b_i\}$, by expected revenue 
$\{b_iw_i\}$, and the 
\textit{squashed} ranking $\{b_iw_i^\a\}$ \cite{lahaie}, all with a
possible reserve score (i.e.\ a per-impression reserve).
It also incorporates our proposed algorithm
$\{(b_i-r)w_i\}$ with reserve price~$r$ (i.e.\ a per-click reserve).

Note, however, that the standard ranking algorithm
$\{b_iw_i\}$ with reserve price~$r$ corresponds to the ranking
function $z(b,w)=\I\{b\geq r\}\,bw$, which is \emph{not} of the
required form.
This introduces some analytical complexities later when we wish to
compare the properties of our proposed algorithm to this algorithm. 
While Theorem~\ref{thm:lowest} guarantees SNE of ranking
algorithms in the class~\eqref{e-y} are well behaved, the same cannot be
said of the standard ranking with a reserve price.  In fact, we will see
that this algorithm can be quite poorly behaved, in a sense that will be
made clear later.

We conclude this section with an additional justification of our
focus on the lowest equilibrium. This justification has the additional benefit
of applying even for rankings outside of the class~\eqref{e-y}, for which SNE may not exist,
a feature we exploit in the next section.
Edelman and Schwarz~\shortcite{edelman2}
argue that because SNE is a full information solution
concept used to model the outcome of a game that is in reality one of
incomplete information, one should only consider SNE
that are in some sense ``feasible'' in the Bayesian
setting. They defined what they called the
\textit{Non-Contradiction Criteria} (NCC), which deems a SNE
implausible if it generates greater expected revenue than any BNE
of the corresponding repeated game of incomplete information.  Rather than
characterising the BNE of the repeated game, they use the
revenue of the optimal BNE as an upper bound.
In their setting, this upper bound on revenue
exactly matches the revenue of the lowest SNE,
and therefore they argue it is the only reasonable
equilibrium.

In our work, we are interested in understanding the behaviour of a
large class of ranking algorithms, none of which need be optimal.
So in our setting, the revenue of the optimal BNE, while still
an upper bound, does not necessarily match the revenue
given by the lowest SNE of an arbitrary ranking algorithm of the form \eqref{e-y}.
However, we know from Theorem~\ref{thm:lowest} that given a ranking
function $y$, all SNE share the same allocation
rule~$\arule^y(\bt,\bw)$.  Therefore, a natural comparison is to
BNE that also share the same allocation rule.  Following Edelman and Schwartz,
rather than characterising such equilibria, we instead derive an upper
bound on their revenue. Since we have fixed the allocations,
Myerson's theory allows us to trivially derive such an upper bound.

\begin{proposition}\label{prop-NCC}
Given a ranking function $y$, the optimal BNE that
ranks ads by $y(\theta_i,w_i)$ has expected revenue 
\[
R(\arule^y)= \E\biggl[\sum_{i=1}^n\phi(\t_i)w_i\arule^y_i(\bt,\bw)\biggr]\,.
\]
\end{proposition}

By Theorem~\ref{thm:lowest}, this upper bound exactly matches the
revenue of the lowest SNE, providing
additional justification for our decision to use it as a focal outcome
of a GSP auction.
Further, any method of selecting an SNE given
types  $(\bt,\bw)$ implicitly defines such a ranking function $y$, not
necessarily within the class \eqref{e-y}, so this upper bound
remains useful even for ranking algorithms outside this class.

\section{Revenue Dominance}
\label{s-revenue}


In this section we compare the revenue
generated by our proposed
ranking algorithm $\{(b_i-r)w_i\}$ with the standard ranking $\{b_iw_i\}$,
both employing the same per-click reserve price $r$. This comparison is of
particular interest as the two algorithms exclude the same set of
advertisers, thus isolating the effect of incorporating the reserve price
into the ranking function. We find that for sufficiently
small\footnote{We give a detailed discussion on the meaning of
  ``sufficiently small'' later, however most reasonable reserve levels
do suffice.} reserve prices,
our proposed algorithm is guaranteed to generate greater
revenue. While in practice the designer may not be solely interested
in revenue, this result helps to show how our algorithm may offer favourable
tradeoffs between revenue and welfare. That is, for a given target
revenue, a designer using our ranking algorithm needs to use a
smaller (and thus less distortionary) reserve price than a designer
employing the standard ranking.

We take our assumption of equilibrium behaviour in GSP auctions to be SNE
whose revenue does not exceed the bound in Proposition~\ref{prop-NCC}.
For our proposed ranking algorithm $\{(b_i-r)w_i\}$
this is equivalent to taking the lowest SNE.
However, the standard ranking algorithm $\{b_iw_i\}$ with
reserve price $r$ has the corresponding ranking function
$z(b,w)=\I\{b\geq r\}\,bw$, which is not within the class
\eqref{e-y}.  With this ranking algorithm, we are not certain whether or not  a SNE
is guaranteed to exist.  In Appendix~\ref{a-SNEcounter}, we give an
example where we can show that any SNE that does exist cannot always rank
ads by $\theta_i w_i$.  That is, ads do not necessarily appear in
the desired order.

Despite the complexity of behaviour with this ranking algorithm, we present the
following theorem which states that, for sufficiently small reserve prices,
the lowest SNE of the GSP auction subject to
our proposed ranking $\{(b_i-r)w_i\}$ generates greater expected
revenue than \textit{any} SNE under the standard ranking
$\{b_iw_i\}$ (with the same reserve price $r$) that respects the
revenue upper bound from Proposition~\ref{prop-NCC}.  

\begin{theorem}\label{t-main}
For $r\in(0,\bart]$,\footnote{Recall that $\bart$ is the theoretical
  revenue-optimal reserve price --- that is, the value at
  which an advertiser's virtual value becomes zero: $\phi(\bart)=0$.} define $R_1(r)$ and
$R_2(r)$ to be the expected revenues from two allocation rules
that select outcomes that are SNE and do not exceed the bound from Proposition~\ref{prop-NCC}
under the ranking algorithms $\{(b_i-r)w_i\}$ and $\{b_iw_i\}$
respectively. If
\beq\label{e-maincond}
r\leq\inf_{t \geq r}\Bigl\{t-\frac{\phi(t)}{\phi'(t)}\Bigr\}\,,
\eeq
then $R_1(r)> R_2(r)$.
\end{theorem}

Informally, condition \eqref{e-maincond} seems to hold for most
reasonable distributions and for most $r\in(0,\bart]$. More precisely,
it will  hold for all $r\in(0,\bart]$ when $\phi(\t_i)$ is weakly convex.  
It is  straightforward to show that  sufficient conditions for $\phi(\t_i)$ to be weakly convex are that $f$ is log-concave and non-increasing.
Log-concavity is a property of many common distributions and is a standard assumption in economic analysis~\cite{bagnoli}.  Requiring $f$ to be non-increasing is somewhat restrictive, but permits, for example, the uniform or exponential distribution.  
Conversely, If $\phi(\t_i)$ is concave then it is likely
that condition~\eqref{e-maincond} does not hold for some choices of
$r$. For example, consider $\t_i\sim{\sf Beta}(2,2)$ which has a
monotone hazard rate and defines $\phi(\t_i)$ to be concave. In this case
$\bart=0.4215$, and for all $r$ the RHS of \eqref{e-maincond} is
minimised at $t=1$ to the value $\frac{1}{3}$. Thus, there exists
an interval $(\frac{1}{3},\bart]$ in which $r$ does not satisfy
condition~\eqref{e-maincond}.   

For choices of $r$ that do not satisfy
condition~\eqref{e-maincond}, it does not follow that our ranking
algorithm therefore generates less revenue than the standard. On the
contrary, we expect our ranking algorithm to generate more revenue in
most cases. 
To give an intuitive explanation, our proof of
Theorem~\ref{t-main} involves showing that one can apply a large
number of pairwise allocation swaps to transform the
allocation rule arising from the standard ranking to that of our
proposed ranking, each of which increases revenue. If $r$ is slightly
greater than the RHS of
\eqref{e-maincond} then a small proportion of swaps will decrease
revenue, while most will still cause an increase. In many such cases, the
net result will still be a revenue increase.  Indeed, we see such
behaviour in our simulations given below.


We work up to Theorem~\ref{t-main} through a series of lemmas.
As previously mentioned, the upper bound from Proposition~\ref{prop-NCC}
is well defined for arbitrary monotone allocation rules.
Let $R(\arule)$ be the value of this bound for the allocation rule
$\arule$:
\beq\label{e-R3}
R(x)=\E\biggl[\sum_{i=1}^n\phi(\t_i)w_i\arule_i(\bt,\bw)\biggr]\,.
\eeq
If $\arule$ is not monotone then Proposition~\ref{prop-NCC}
no longer holds, however $R(\arule)$ is still well
defined as the same functional form. The only difference in this case
is that $R(\arule)$ does not translate as an achievable revenue.

Our first lemma shows how one can increase the integrand
of \eqref{e-R3} for a given realisation $(\bt,\bw)$. This is achieved by
performing a simple adjustment or \emph{swap} to the allocation rule
$\arule$.

\begin{lemma}\label{l-swap}
Suppose $\arule$ is an allocation rule
for which there exists a realisation $(\bt,\bw)$ and specific $i,j$ such that 
\begin{align*}
\phi(\t_i)w_i&>\phi(\t_j)w_j\,,\\
\arule_i(\bt,\bw)&< \arule_j(\bt,\bw)\,.
\end{align*}
Define the adjusted allocation rule $\tilde{\arule}$ which is
identical to $\arule$ except for the single swap
$\tilde{\arule}_i(\bt,\bw)=x_j(\bt,\bw)$ and vice versa.
Then, 
\[
\sum_{i=1}^n\phi(\t_i)w_i\tilde{\arule}_i(\bt,\bw)> \sum_{i=1}^n\phi(\t_i)w_i\arule_i(\bt,\bw)\,.
\]
\end{lemma}

\begin{proof}
Direct from the conditions.
\qed\end{proof}


Our next lemma follows as a corollary to Lemma~\ref{l-swap}, extending
it to situations where improvements are possible through a sequence of
swaps.

\begin{lemma}\label{l-yz}
Let $\arule^y$ and $\arule^z$ be two allocation rules such that
the following properties hold for all $\t_i,w_i,\t_j,w_j$:
\beq\label{e-yz1}
y(\t_i,w_i)=0\quad\Leftrightarrow\quad z(\t_i,w_i)=0\,,
\eeq
\beq
\Bigl\{y(\t_i,w_i)> y(\t_j,w_j)\quad\text{AND}\quad z(\t_i,w_i)<
z(\t_j,w_j)\Bigr\}\quad\Rightarrow\quad\phi(\t_i)w_i>\phi(\t_j)w_j\,. \label{e-yz2}
\eeq
Then, $R(\arule^y)\geq R(\arule^z)$. Furthermore, if $\arule^y$ and
$\arule^z$ differ with positive probability then $R(\arule^y)> R(\arule^z)$.
\end{lemma}

\begin{proof}
The intuition behind Lemma~\ref{l-yz} is clear --- if it holds that
any time $y$ and $z$
disagree about the ranking of two advertisers then $y$ is
`correct', then it should hold that $R(\arule^y)>
R(\arule^z)$. The proof involves
showing that one can perform a sequence of swaps to transform $\arule^z$
into $\arule^y$, where each swap satisfies the conditions of
Lemma~\ref{l-swap} (see Appendix~\ref{a-yz}).
\qed\end{proof}

Our third lemma is the heart of the proof. It shows that for a
sufficiently small reserve price $r$, the two ranking functions
defined by our proposed algorithm and the standard satisfy conditions
\eqref{e-yz1} and \eqref{e-yz2} from the previous lemma.

\begin{lemma}\label{l-Rdom}
Given a reserve price $r\in(0,\bart]$, let
\begin{align}
&y(\t_i,w_i)=(\t_i-r)^+w_i\,,\label{e-Rdomy}\\
&z(\t_i,w_i)=\I\{\t_i\geq r\}\,\t_iw_i\,.\label{e-Rdomz}
\end{align}
If
\beq\label{e-Rdomr}
r\leq\inf_{t \geq r}\Bigl\{t-\frac{\phi(t)}{\phi'(t)}\Bigr\}\,,
\eeq
then $R(\arule^y)> R(\arule^z)$.
\end{lemma}

\begin{proof}
The proof of Lemma~\ref{l-Rdom} involves showing that conditions
\eqref{e-yz1} and \eqref{e-yz2} are satisfied (see Appendix~\ref{a-Rdom}).
\qed\end{proof}

The only remaining technical detail is that the ranking function
$z(b,w)=\I\{b\geq r\}\,bw$ is not within the class
\eqref{e-y}.  As previously discussed, this means that we need to consider
the possibility that there may exist SNE with a different ranking from
$z(\t_i,w_i)$, to which we cannot directly apply the upper bound
$R(\arule^z)$.  However, our final lemma shows that any such alternate
rankings can only reduce our upper bound on revenue.


\begin{lemma}\label{l-zSNE}
Let $\arule$ be an allocation rule that selects a SNE of a
GSP auction with the ranking function $z(b,w)=\I\{b\geq
r\}\,bw$.  Then $R(\arule^z)\geq R(\arule)$.
\end{lemma}


\begin{proof}
See Appendix~\ref{a-zSNE}.
\qed\end{proof}

The addition of Lemma~\ref{l-zSNE} is sufficient to complete our proof of
Theorem~\ref{t-main}. By Theorem~\ref{thm:lowest},
$R_1(r)=R(\arule^y)$; by Lemma~\ref{l-Rdom},
$R(\arule^y)> R(\arule^z)$; and by Lemma~\ref{l-zSNE},
$R(\arule^z)\geq R_2(r)$. Thus we have $R_1(r)> R_2(r)$. 

\section{Simulations}
\label{s-example}

In this section, we use simulations to examine the performance of our ranking
algorithm and show that it generally dominates existing ranking algorithms.
We examine three metrics: revenue, welfare, and click yield.
Revenue is what the auctioneer cares about (at least in the short term).
Welfare is the total value created for advertisers ($\sum \t_iw_ix_i$), and
the auctioneer also cares about this for the long-term health of the platform.
Similarly, click yield (i.e. the total number of clicks $\sum w_ix_i$) can be
thought of as a proxy for the value created for the users who are clicking
on (presumably) useful ads.

There is one technical detail relevant to
Figs.~\ref{f-ExRevDom}-\ref{f-LahaieClickYield}. As previously
discussed, the standard ranking algorithm coupled with a
reserve price $r$ corresponds to the ranking function
$z(b,w)=\I\{b\geq r\}\,bw$, which is not within the class \eqref{e-y}.
As a consequence, any existing SNE may not be well-behaved. Instead of
trying to characterise such equilibria, we use the relevant
statistics of the optimal BNE which ranks ads by $z(\t_i,w_i)$. By
Lemma~\ref{l-zSNE}, the BNE revenue $R(\arule^z)$ is an
upper bound for the corresponding SNE revenue. Thus, the (red) curves
may display overestimates of the true revenues.

The figures' legends in this section use a simple notational
shorthand to signify the type of reserve --- that is, $r$ denotes a
reserve price and $\rho$ a reserve score (i.e.\ a per-impression reserve).

\begin{figure}[ht]
\begin{minipage}[t]{0.48\linewidth}
\centering
\psfrag{r}{\scriptsize{$r$}}
\psfrag{longexpression2}{\scriptsize{$\{b_iw_i\}$ / $r$}}
\psfrag{longexpression1}{\scriptsize{$\{(b_i-r)w_i\}$}}
\includegraphics[width=1.05\linewidth]{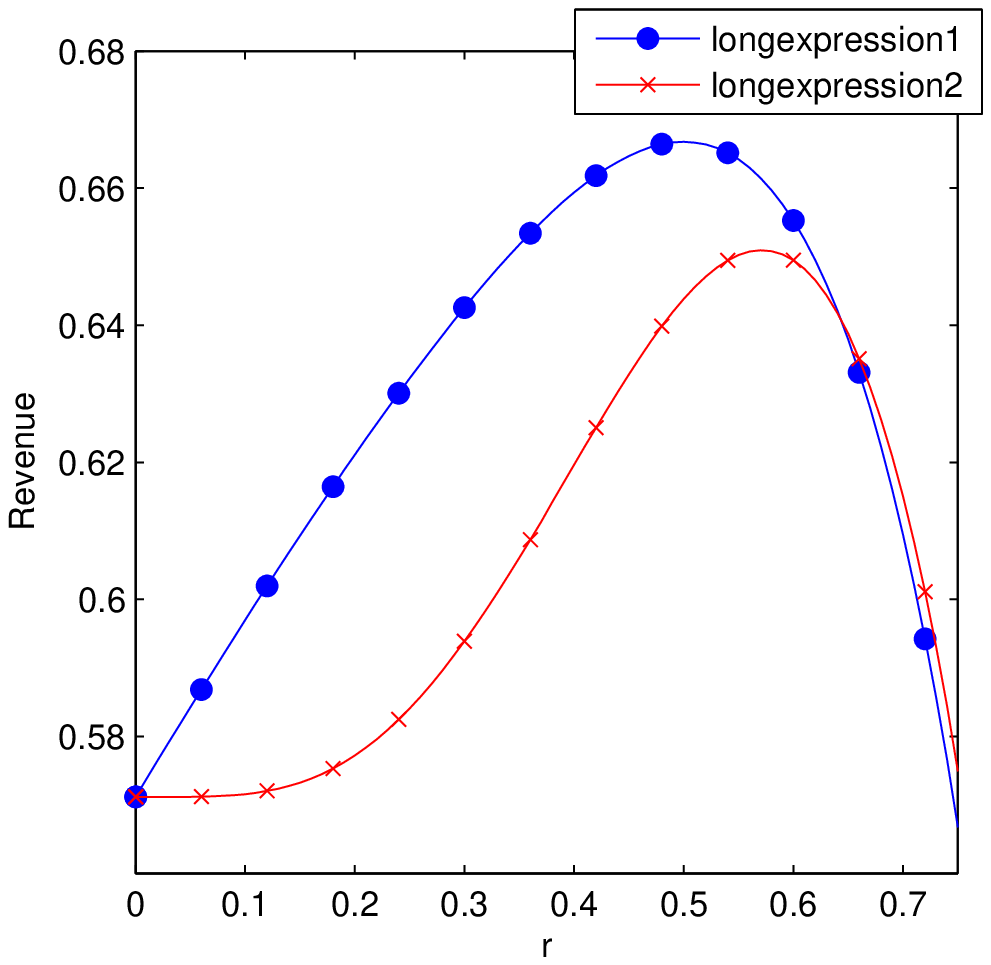}
\caption{\small{Revenue comparison between our proposed algorithm and the
  standard ranking in a simple setting.}}
\label{f-ExRevDom}
\end{minipage}
\hspace{3mm}
\begin{minipage}[t]{0.48\linewidth}
\centering
\psfrag{longexpression3}{\scriptsize{$\{b_iw_i^\a\}$ / $\rho$}}
\psfrag{longexpression1}{\scriptsize{$\{(b_i-r)w_i\}$}}
\psfrag{longexpression2}{\scriptsize{$\{b_iw_i\}$ / $r$}}
\psfrag{longexpression4}{\scriptsize{$\{b_iw_i^\a\}$}}
\psfrag{longexpression5}{\scriptsize{$\{b_iw_i\}$ / $\rho$}}
\includegraphics[width=1.06\linewidth]{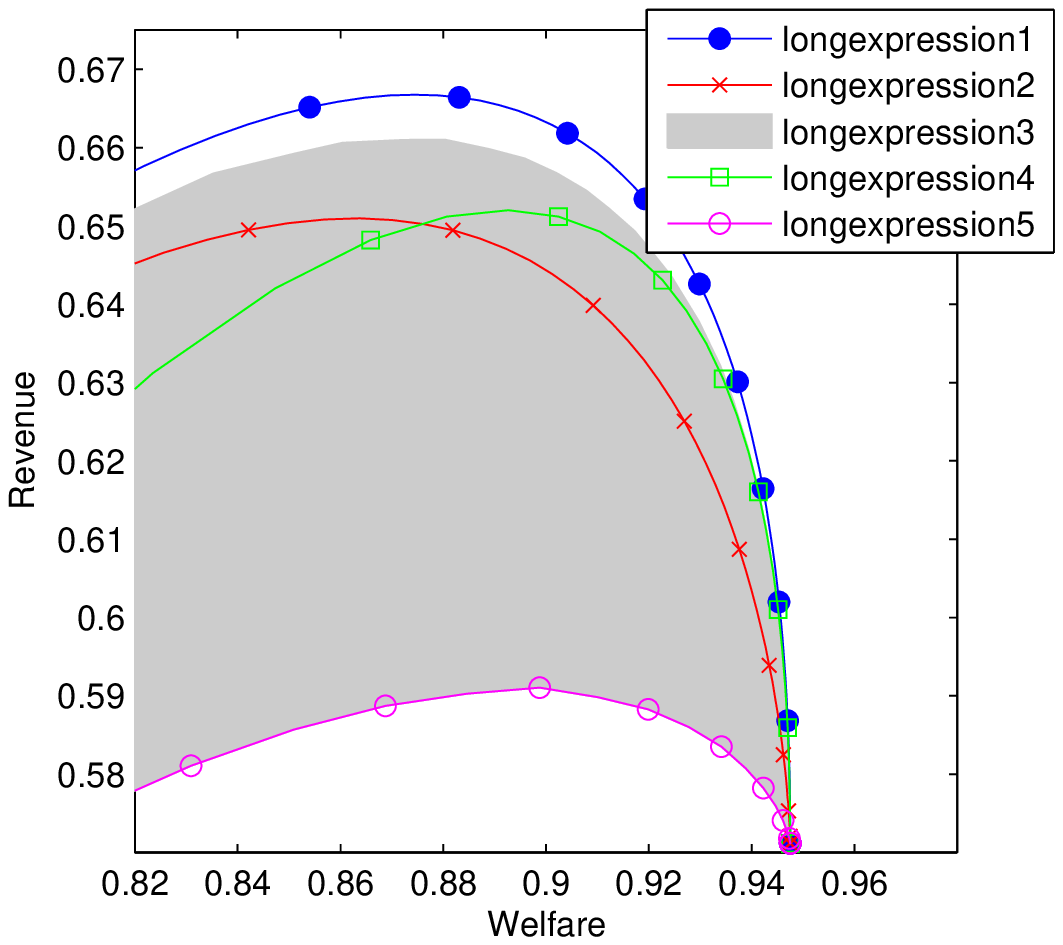}
\caption{\small{Feasible welfare-revenue operating points in a simple setting.}}
\label{f-ExWelfare}
\end{minipage}
\end{figure}

We begin with a simple (albeit unrealistic) example, which satisfies
the distributional assumptions made in Sect.~\ref{s-model}. There are 
eight advertisers bidding for three slots. Advertisers have
i.i.d.\ types $(\t_i,w_i)$ where $\t_i$ and $w_i$ are independent and both
uniformly distributed on $[0,1]$. Figure~\ref{f-ExRevDom} illustrates
Theorem~\ref{t-main} in this setting: for all $r\leq 0.5\, (=\bart)$, our
proposed algorithm of incorporating the reserve price into the ranking
function raises more revenue than the standard ranking.
In this simple setting, we can actually achieve the optimal revenue at
$r=0.5$.

However, Theorem~\ref{t-main} does not tell us what the cost of this
added revenue is in terms of welfare.  Figure~\ref{f-ExWelfare}
shows us that this revenue is essentially free: for any welfare
we desire, we can achieve more revenue with our ranking algorithm.
Note that this does not mean that with the same reserve price
our algorithm is more efficient. Instead, if separate reserve prices are chosen
such that both algorithms have the same welfare, our algorithm has higher
revenue.

In Fig.~\ref{f-ExWelfare}, we also compare performance against
Lahaie and Pennock's~\shortcite{lahaie} squashed ranking algorithm with reserve score
$\rho$ (i.e.\ $y(b,w)=(bw^\a-\rho)^+$).  Since there are two
parameters, the operating points form the entire shaded region.
Our algorithm leads to a set of operating points that dominates this algorithm
as well. Two special cases of this ranking algorithm are 
the squashing ranking algorithm with no reserve ($\rho = 0$) and the standard
ranking algorithm with a reserve score ($\alpha = 1$). The latter is particularly
interesting to compare to the standard ranking algorithm with a reserve
price. We observe that for identical pre-reserve rankings, the addition of a reserve
price dominates the alternative option of a reserve score. Despite the
fact that the plotted revenues of the 
standard ranking with a reserve price may be overestimates, this
still suggests that it is generally better to use reserve prices
than reserve scores. Figure~\ref{f-ExClickYield} shows that these results do
not change if we examine click yield rather than welfare.

\begin{figure}[ht]
\begin{minipage}[t]{0.48\linewidth}
\centering
\psfrag{longexpression3}{\scriptsize{$\{b_iw_i^\a\}$ / $\rho$}}
\psfrag{longexpression1}{\scriptsize{$\{(b_i-r)w_i\}$}}
\psfrag{longexpression2}{\scriptsize{$\{b_iw_i\}$ / $r$}}
\psfrag{longexpression4}{\scriptsize{$\{b_iw_i^\a\}$}}
\psfrag{longexpression5}{\scriptsize{$\{b_iw_i\}$ / $\rho$}}
\includegraphics[width=1.06\linewidth]{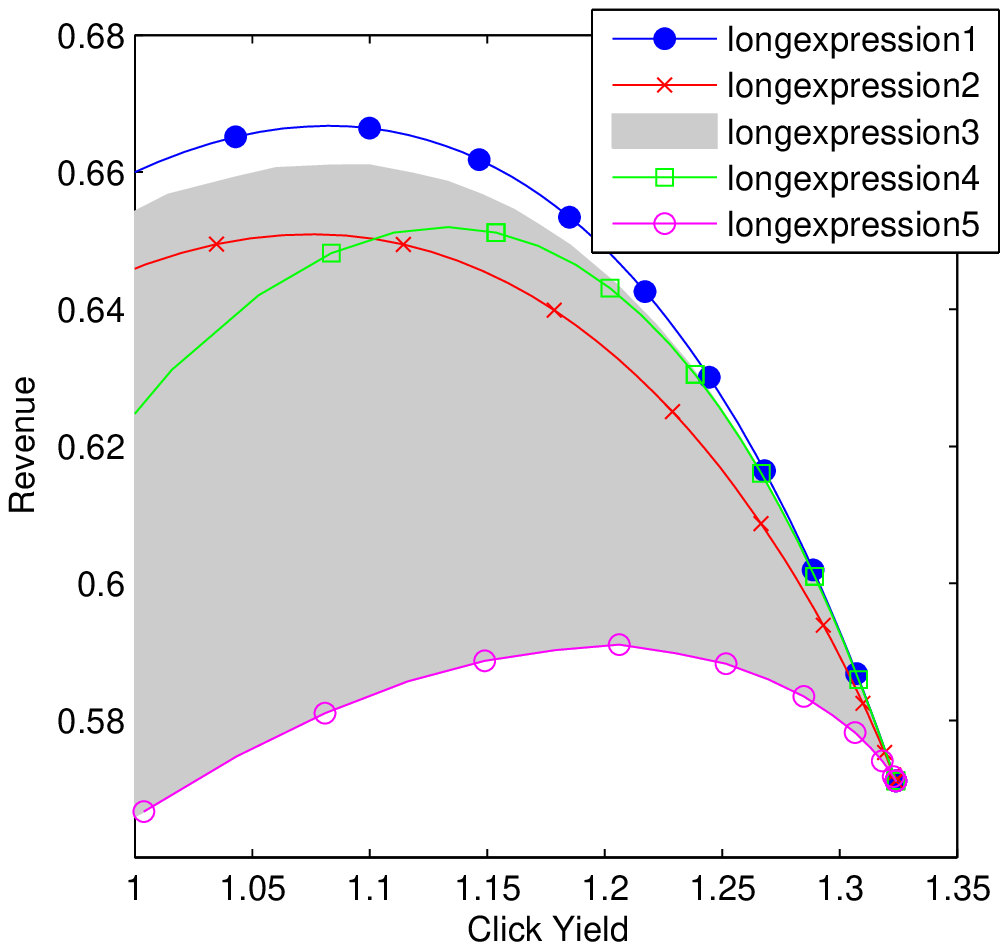}
\caption{\small{Feasible click yield-revenue operating points in a simple setting.}}
\label{f-ExClickYield}
\end{minipage}
\hspace{3mm}
\begin{minipage}[t]{0.48\linewidth}
\centering
\psfrag{r}{\scriptsize{$r$}}
\psfrag{longexpression2}{\scriptsize{$\{b_iw_i\}$ / $r$}}
\psfrag{longexpression1}{\scriptsize{$\{(b_i-r)w_i\}$}}
\includegraphics[width=1.05\linewidth]{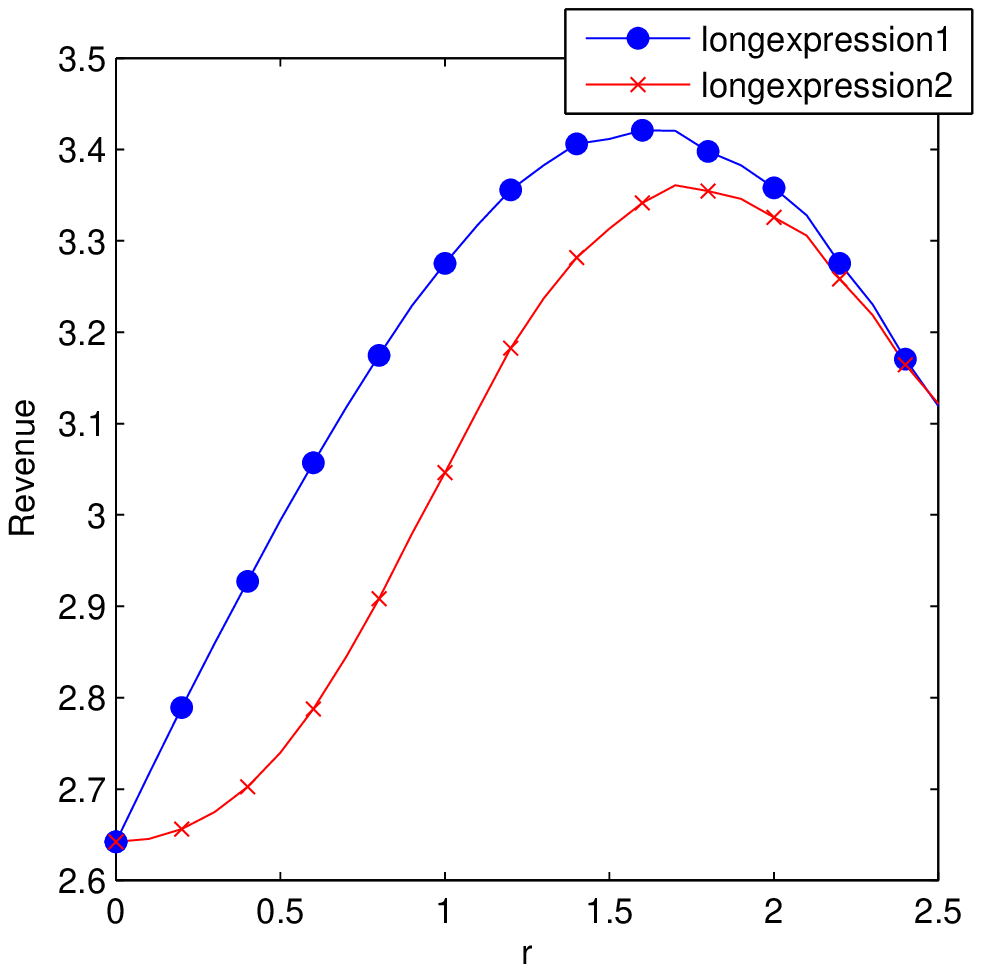}
\caption{\small{Revenue comparison between our proposed algorithm and the
  standard ranking in Lahaie and Pennock's setting.}}
\label{f-LahaieRevDom}
\end{minipage}
\end{figure}

Lahaie and Pennock~\shortcite{lahaie} examined the performance of the squashed ranking in a
more realistic setting, which they selected by fitting Yahoo! data from a particular query.
This distribution violates several of our assumptions. Bidder valuations have a lognormal
distribution, which does not have a monotone hazard rate.  Values are also correlated with
relevance.  Nevertheless,
Figs.~\ref{f-LahaieRevDom}-\ref{f-LahaieClickYield} show that the
results from our simple setting are essentially unchanged, with our
proposed algorithm of incorporating a reserve price into the ranking function
offering superior tradeoffs.

\begin{figure}[ht]
\begin{minipage}[t]{0.48\linewidth}
\centering
\psfrag{longexpression3}{\scriptsize{$\{b_iw_i^\a\}$ / $\rho$}}
\psfrag{longexpression1}{\scriptsize{$\{(b_i-r)w_i\}$}}
\psfrag{longexpression2}{\scriptsize{$\{b_iw_i\}$ / $r$}}
\psfrag{longexpression4}{\scriptsize{$\{b_iw_i^\a\}$}}
\psfrag{longexpression5}{\scriptsize{$\{b_iw_i\}$ / $\rho$}}
\includegraphics[width=1.07\linewidth]{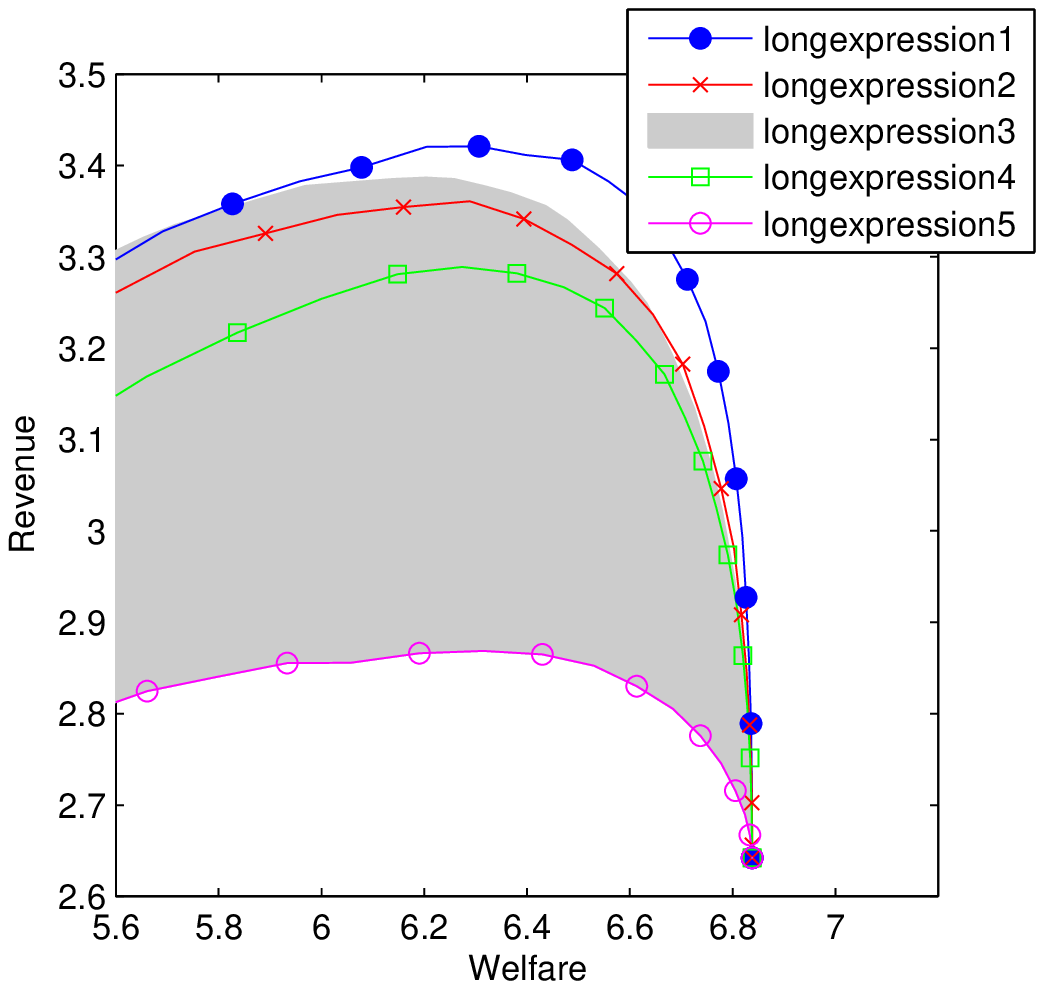}
\caption{\small{Feasible welfare-revenue operating points in Lahaie and Pennock's setting.}}
\label{f-LahaieWelfare}
\end{minipage}
\hspace{3mm}
\begin{minipage}[t]{0.48\linewidth}
\centering
\psfrag{longexpression3}{\scriptsize{$\{b_iw_i^\a\}$ / $\rho$}}
\psfrag{longexpression1}{\scriptsize{$\{(b_i-r)w_i\}$}}
\psfrag{longexpression2}{\scriptsize{$\{b_iw_i\}$ / $r$}}
\psfrag{longexpression4}{\scriptsize{$\{b_iw_i^\a\}$}}
\psfrag{longexpression5}{\scriptsize{$\{b_iw_i\}$ / $\rho$}}
\includegraphics[width=1.05\linewidth]{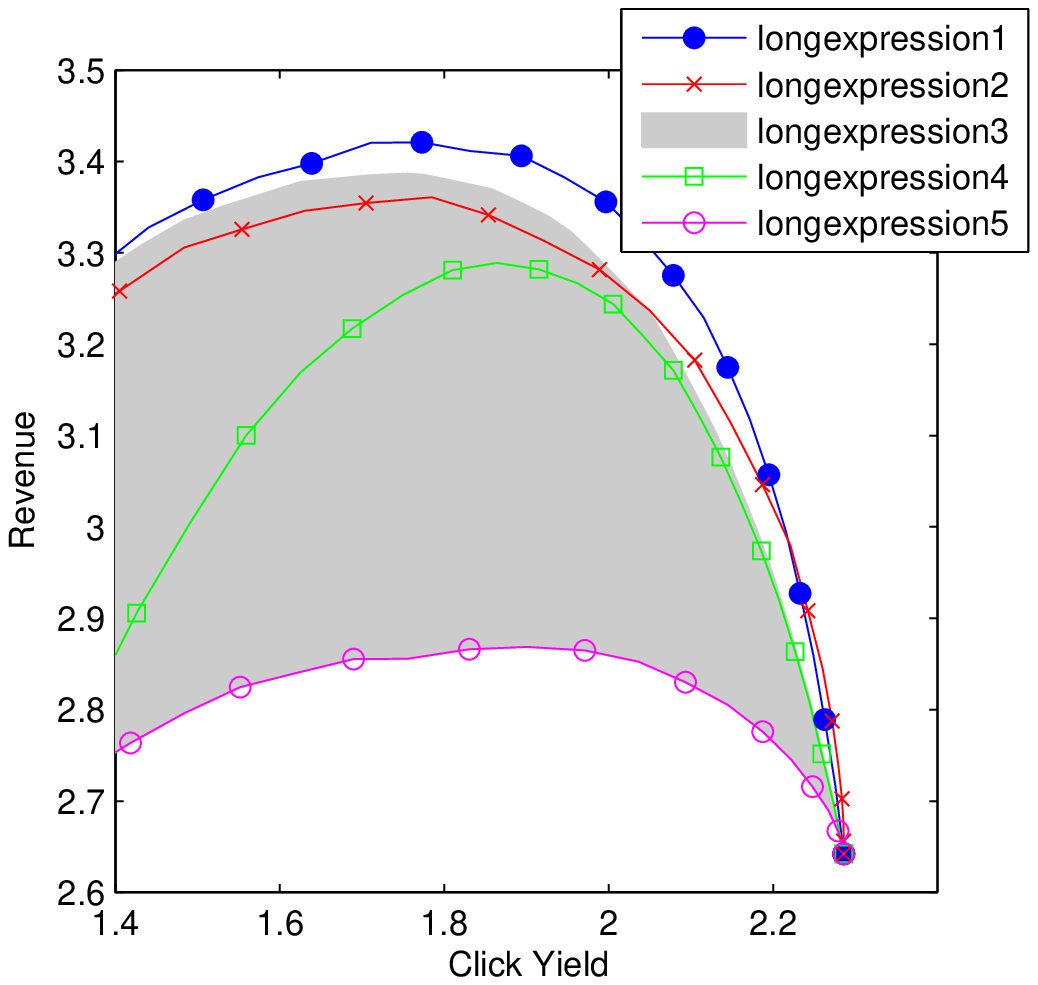}
\caption{\small{Feasible click yield-revenue operating points in Lahaie and Pennock's setting.}}
\label{f-LahaieClickYield}
\end{minipage}
\end{figure}

Finally, all of these results are based on the assumption that bidders are
in equilibrium.  In reality, if parameters are changed, the algorithm may take
some time to reach the new equilibrium, and there is empirical evidence that some
advertisers react quite slowly to changes.  Therefore, a natural
question is what happens when the ranking algorithm is changed but
advertisers do not react?  If the short-term effect is revenue-positive or
revenue-neutral, it is much easier for the auctioneer to be patient.
Furthermore, by not requiring an equilibrium analysis, we can examine
the performance of different ranking algorithms on historical data,
which has many realistic features not captured by our simple model
(e.g.\ changing bidders, matching of bids to multiple queries, and
stochastic quality scores).

Figure~\ref{f-thick} shows the effect on revenue of changing from
from the standard ranking algorithm to our proposed ranking algorithm
while keeping the reserve price fixed on Microsoft historical data
for a keyword with over 500 bidders, which we have selected as
representative of a ``thick'' market.  The data has been normalised, but
the exact values are not relevant for our purposes.  In such markets,
incorporating the reserve price into the ranking function seems to
consistently increase revenue.

\begin{figure}[ht]
\begin{minipage}[t]{0.48\linewidth}
\centering
\psfrag{r}{\scriptsize{$r$}}
\psfrag{R}[0][0][1][270]{\hspace{-2mm}\small{$R$}}
\psfrag{longexpression2}{\scriptsize{$\{b_iw_i\}$ / $r$}}
\psfrag{longexpression1}{\scriptsize{$\{(b_i-r)w_i\}$}}
\includegraphics[width=1.05\linewidth]{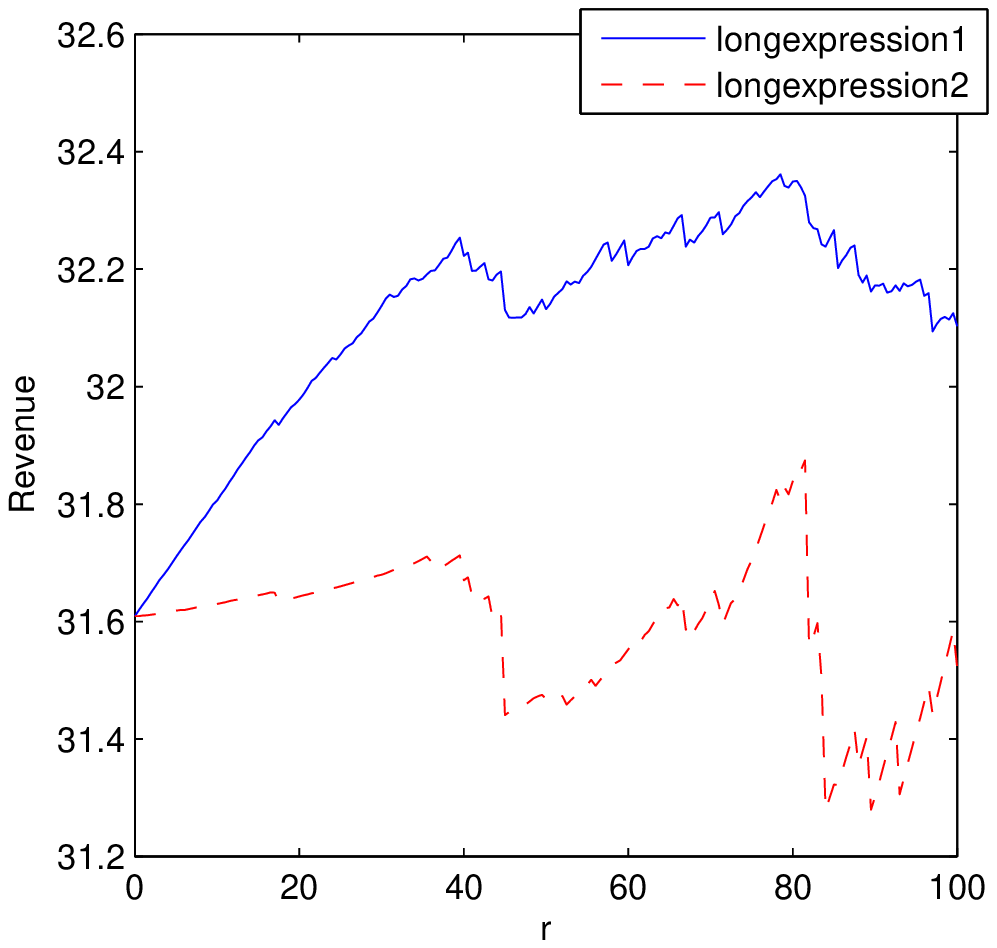}
\caption{\small{Revenue comparison between our proposed algorithm and the
  standard ranking in thick market auction replays.}}
\label{f-thick}
\end{minipage}
\hspace{3mm}
\begin{minipage}[t]{0.48\linewidth}
\centering
\psfrag{r}{\scriptsize{$r$}}
\psfrag{longexpression2}{\scriptsize{$\{b_iw_i\}$ / $r$}}
\psfrag{longexpression1}{\scriptsize{$\{(b_i-r)w_i\}$}}
\psfrag{longexpression3}{\scriptsize{Bid hist.}}
\includegraphics[width=1.05\linewidth]{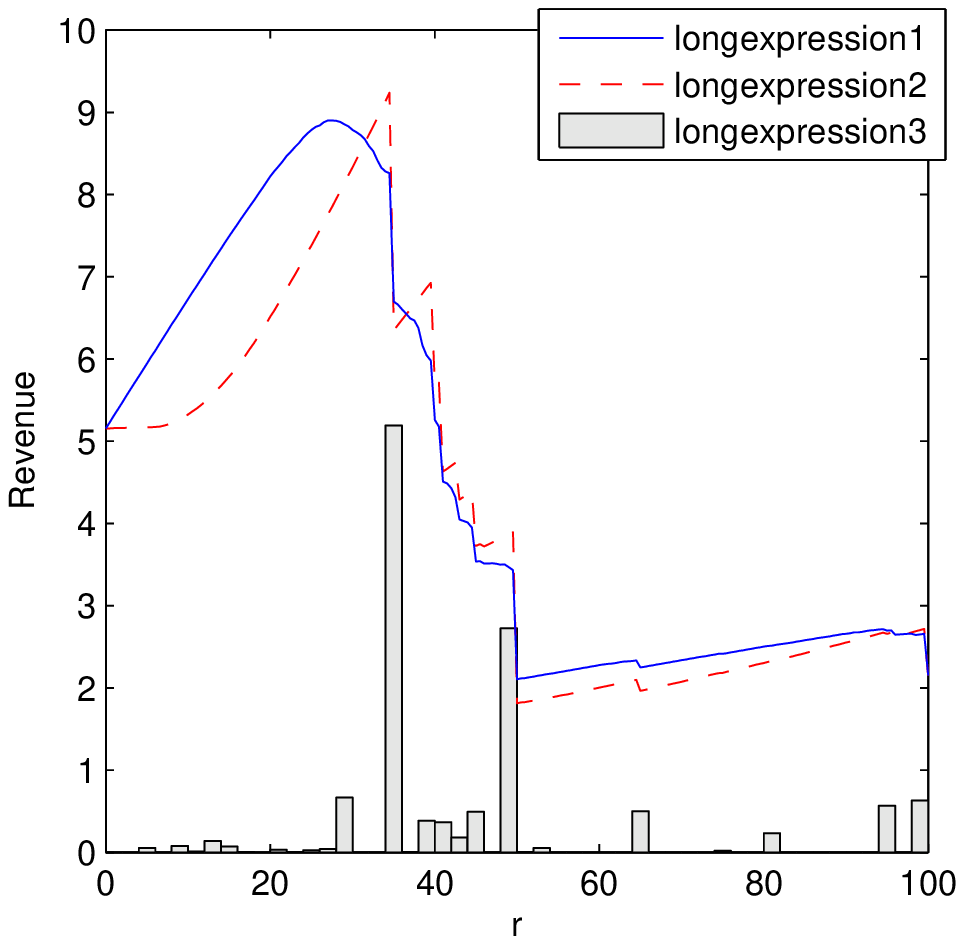}
\caption{\small{Revenue comparison between our proposed algorithm and the
  standard ranking in thin market auction replays.}}
\label{f-thin}
\end{minipage}
\end{figure}

Figure~\ref{f-thin} is a similar plot of a ``thin'' market with fewer
than 10 bidders.  Here, at certain values, the standard ranking
raises somewhat more revenue.  The included histogram
of bid frequencies suggests an explanation for this: setting the
reserve price at a common bid makes those bidders pay their full
value, while the standard ordering ranks them highly to extract
as much revenue as possible.  In practice, such reserve prices
are unlikely to be chosen, as setting a reserve price to match
common bids would essentially make that auction first price, as
well as being very sensitive to small changes in bid.  At more
reasonable choices of reserve price, our proposed algorithm of
incorporating it into the ranking function yields greater revenue.

Figures~\ref{f-thick} and~\ref{f-thin} also demonstrate several advantages of our ranking algorithm from an optimisation perspective.  First, the blue lines are ``smoother'', which creates a somewhat easier problem.  Second, the fact that bidders near the reserve price have low rank scores means that the revenue from an advertiser begins to decrease before the reserve price is actually raised past his bid.  This reduces the tendency of optimisation to overfit and choose a reserve price directly below an advertiser's bid.

\begin{figure}[ht]
\centering
\psfrag{longexpression1}{\scriptsize{$\{(b_i-r)w_i\}$}}
\psfrag{longexpression2}{\scriptsize{$\{b_iw_i\}$ / $r$}}
\psfrag{longexpression5}{\scriptsize{$\{b_iw_i\}$ / $\rho$}}
\includegraphics[width=0.7\linewidth]{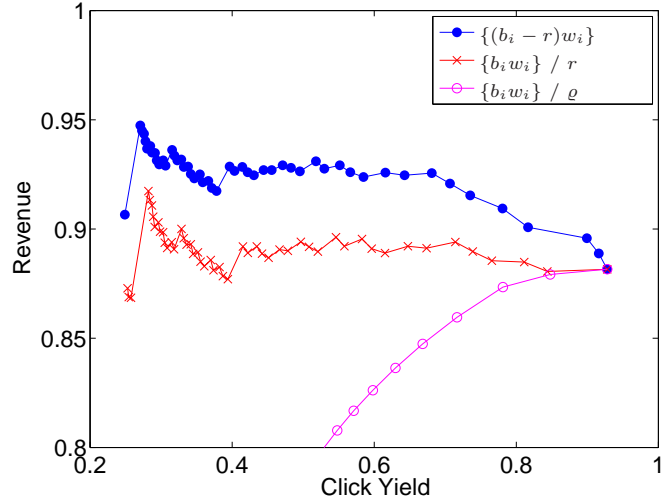}
\caption{\small{Feasible click yield-revenue operating points on Microsoft data.}}
\label{f-Whole}
\end{figure}

To examine the tradeoff between revenue and click yield, we tested a
subset of the ranking algorithms using global parameter settings on a
sample of a week's worth of Microsoft data across all queries.  As there is
a minimum bid of 5 cents in the actual system, an implicit reserve
price of 5 cents is applied to all ranking algorithms in addition to
any other parameters.  Figure~\ref{f-Whole} shows that incorporating the reserve
price into the ranking function results in a better tradeoff than
using the reserve price solely as a minimum bid. As in the single
query case, the standard ranking experiences bigger peaks and drops as
the reserve price approaches common bids (in this figure increasing
the reserve corresponds to moving right to left). Interestingly, a reserve score does not generate a
useful tradeoff  (at least when set globally) as increasing it reduces
both revenue and clicks.  Any gains from raising the price the last ad
shown pays are more than offset by the lower number of clicks.  We
cannot investigate the tradeoff between revenue and welfare as we do
not know the true values of advertisers.

\section{Related Work}

A number of papers have examined auctions other than GSP designed for
sponsored search auctions.
Abrams and Ghosh~\shortcite{abrams}
considered the design of auctions that guarantee revenue to be at least some
fixed proportion of the optimal revenue that can 
be generated by an omniscient designer, something
that is not achieved by the GSP auction. Yenmez~\shortcite{yenmez} studied
equilibrium properties of general auction designs, including looking at conditions
under which the VCG outcome is implementable in equilibrium.

GSP auction literature is generally divided between Bayesian and
non-Bayesian analysis. In the former, Gomes and
Sweeney~\shortcite{gomes} pursued existence and 
uniqueness results regarding BNE of the GSP
auction. Pin and Key~\shortcite{pin} and Athey and Nekipelov~\shortcite{athey}
proposed separate models from which advertisers'
values can be inferred from bid data.   Pin and Key derive an advertiser best-response
 in the face of uncertainty in  a repeated auction setting, and make a connection to 
certain scenarios
when the BNE of \cite{gomes} exist;  whereas Athey and Nekipelov
start directly  from a BNE setting.
Both models can be interpreted in a Bayesian setting, 
but differ in the precise information
available to advertisers.
Athey and Nekipelov do consider the the tradeoffs from introducing various changing to the ranking algorithm, but only consider one particular choice of parameters rather than the full range of operating points.
Athey and Ellison~\shortcite{athey2}
examined tradeoffs in the context of a model of consumer search where
reserve prices can actually increase welfare. In their model a reserve
price changes consumer perceptions of an ad's relevance, which can lead
to more overall clicks.
Vorobeychik~\shortcite{voro} presented a
simulation-based technique to approximate a BNE of the GSP auction,
made feasible by restricting the space of possible bidding
strategies. 
Liu and Chen~\shortcite{liu} and Liu et al.~\shortcite{liu2} compare the designs of the welfare-optimal and revenue-optimal auctions in a simplified settings where there are only two types of advertisers, but also provide some insight into settings with more than two types.

Analysis of the GSP auction in the complete information setting was
pioneered independently by Edelman et al.~\shortcite{edelman}, 
Varian~\shortcite{varian}, and Aggarwal et al.~\shortcite{aggarwal06}.
All three papers note the striking equivalence result
between the VCG outcome and one specific ex-post equilibrium of the
GSP auction --- that is, the `lowest SNE' or `bidder-optimal locally
envy free equilibrium'. Ashlagi et al.~\shortcite{ashlagi} showed that this
result holds more generally. That is, for many auction designs in
which advertisers' payments are calculated as functions of lower-ranked bids, 
there exists a specific equilibrium with payoffs equivalent to the VCG
outcome. Paes Leme and Tardos~\shortcite{paesleme} note that the set of
ex-post equilibria of the GSP auction is larger than the set of SNE, and may contain
inefficient equilibria. For both the complete-information and Bayesian
settings, they investigated the maximum possible loss of
welfare --- that is, they found upper bounds on the \emph{price of
  anarchy}. Their results have since been strengthened in a line of work
that includes~\cite{lucier,cara,cara2}.

Edelman and Schwarz~\shortcite{edelman2}
examined the revenue of various SNE and argue that the `lowest'
is the only reasonable one based on a criterion that compares
SNE revenue to the optimal revenue possible in the Bayesian
setting. They use the resulting insight
to argue that that reserve prices raise more revenue
through their indirect effects (increasing equilibrium bids)
than through their direct effects (raising prices paid) and that
going from a zero reserve price to the optimal reserve
price leads to only a small welfare loss while significantly
increasing revenue. 

Ostrovsky and Schwarz~\shortcite{ostro} presented the results of a field
experiment aimed at testing the effects of employing
Myerson's~\shortcite{myerson} optimal reserve
price in GSP auctions. Working in conjunction with Yahoo!, they used
historical bid data for a large number of queries to estimate
distributions of advertisers' values and subsequently optimal reserve
prices. After employing the new reserve prices, they observed substantial
increases in revenue.  However, their reserves were implemented as
minimum scores and were not used to change to ordering of ads.

Lahaie and Pennock~\shortcite{lahaie} suggest a different method of
increasing revenue. They found that incorporating a squashing
exponent $\a<1$ into the ranking algorithm $\{b_iw_i^\a\}$ generally has a
positive effect on revenue. They used simulations to show that the
optimal value of $\a$ is sensitive to the correlation between
advertisers' values and relevances. They observed that introducing a
squashing exponent to
increase revenue incurs a smaller welfare loss than that of setting
a reserve score. Thus, they argue their method to be preferential. Further
strengthening this justification, Lahaie
and McAfee~\shortcite{lahaie2} showed that when there is uncertainty in the
estimates of advertisers' relevances, introducing a squashing exponent
can actually \emph{increase} welfare by reducing the weight placed on these uncertain estimates.
While we have not explored this issue for our proposed ranking algorithm, it certainly has the same effect, which makes it plausible that their conclusion would apply to our algorithm as well.

\section{Conclusion}

In this paper, we have examined the tradeoffs between revenue,
welfare, and click yield enabled by the choice of algorithm to rank
sponsored search ads.
In developing our solution concept, we have extended the connection between the
lowest SNE  and the VCG mechanism \cite{varian,edelman,aggarwal06} to a much wider class of ranking
algorithms. As this class includes inefficient rankings, our extension instead
establishes an equivalence between the lowest SNE and the corresponding
Myerson~\shortcite{myerson} outcome, of which the SNE/VCG
result is a special case.

We then proposed a new class of ranking algorithms that
incorporates a reserve price/minimum bid
into the ranking algorithm that governs the ordering of advertisers
(not just whether they appear at all), a
system which shares several qualitative features with the
revenue-optimal auction. Our proposed ranking essentially combines two methods of
increasing revenue previously discussed in literature --- namely,
squashing~\cite{lahaie} and setting a
reserve~\cite{edelman,ostro}.  Where previously these two methods have
been implemented using  separate parameters, our algorithm achieves the aims of  both with a
single parameter.

We derived conditions under which, for a
fixed reserve price, our proposed ranking algorithm generates greater
revenue than the standard ranking. This comparison is particularly
informative, as it isolates the effect of incorporating the reserve
price into the ranking function. It also provides intuition for why
our ranking enables good tradeoffs: to raise a given amount of
revenue our ranking algorithm can use a lower (less distortionary)
reserve price.
Our theorem work relies on various
distributional assumptions that may not hold in reality, however we
expect the comparison to be very similar in many cases which violate
the underlying assumptions.

We finished with extensive simulations of the tradeoffs enabled by different ranking algorithms.  We used two types of simulation: one simulating an reactive environment where users react to changes in auction design by forming revised equilibria; the other using real auction logs which assumes advertisers do not have time to react to changes, but which captures all the vagaries of real auctions, where bids change in the light of budget constraints and changing users, and where underlying quality factors are stochastic.    These simulations show that our proposed ranking algorithm enables superior tradeoffs using a variety of metrics.  In particular, to achieve a fixed revenue increase, our ranking algorithm incurred a
smaller loss of welfare and click yield than the alternative rankings.
These simulations also showed that our proposed ranking
has several nice properties from an optimisation perspective.

\bibliography{references}

%
%

\newpage
\section*{Appendix}
\appendix

\section{Non-Existence of an Order-Preserving SNE}
\label{a-SNEcounter}
For the standard ranking algorithm with a reserve price (i.e.\
$z(b,w)=\I\{b\geq r\}\,bw$), we show that a SNE cannot always rank ads
by $z(\t_i,w_i)$, in contrast to SNE under ranking
algorithms within the class \eqref{e-y}. The SNE inequalities~\eqref{e-SNEineq}
can be written as
\beq\label{ea-zSNEineq}
\bigl(\t_iw_i-\max\{rw_i,b_{i+1}w_{i+1}\}\bigr)x_i\geq
\bigl(\t_iw_i-\max\{rw_i,b_{j+1}w_{j+1}\}\bigr)x_j\quad\text{for
  all}\,\,i,j\,.
\eeq
Consider the following realisation. There are precisely two
advertisers who submit qualifying bids ($b_i\geq
r$), with bidder $1$ being awarded the top slot and bidder $2$ the second
($b_1w_1\geq b_2w_2$ and $x_1>x_2$). Bidder $1$ is
less relevant ($w_1<w_2$).

Suppose a SNE always ranks ads by $z(\t_i,w_i)$,
so that $\t_1w_1\geq \t_2w_2$. The bids $(b_1,b_2)$ must satisfy the inequalities
\begin{align}
&\bigl(\t_1w_1-\max\{rw_1,b_2w_2\}\bigr)x_1\geq(\t_1w_1-rw_1)x_2\,,\label{ea-SNEcounter1}\\
&(\t_2w_2-rw_2)x_2\geq\bigl(\t_2w_2-\max\{rw_2,b_2w_2\}\bigr)x_1\,.\label{ea-SNEcounter2}
\end{align}
It is necessary that $b_2>r$ in order to satisfy
\eqref{ea-SNEcounter2}, and as $w_1<w_2$, we have
$\max\{rw_1,b_2w_2\}=\max\{rw_2,b_2w_2\}=b_2w_2$. Then inequalities
\eqref{ea-SNEcounter1} and \eqref{ea-SNEcounter2} can be rewritten:
\begin{align}
&b_2w_2x_1\leq\t_1w_1(x_1-x_2)+rw_1x_2\,,\label{ea-ineq3}\\
&b_2w_2x_1\geq\t_2w_2(x_1-x_2)+rw_2x_2\label{ea-ineq4}\,.
\end{align}
We need the RHS of \eqref{ea-ineq3} to be at least as large as the RHS of
\eqref{ea-ineq4}. However, this
is not always the case. For example, suppose $(\t_1,w_1)=(1,0.7)$,
$(\t_2,w_2)=(0.6,1)$, $r=0.5$, and $(x_1,x_2)=(1,0.5)$. We find the
bounds on advertiser $2$'s bid to be $b_2\geq0.55$ and
$b_2\leq0.525$. Thus, a SNE under the standard ranking algorithm with a
reserve price does not necessarily rank ads by $\t_iw_i$.

\section{Proof of Lemma~\ref{l-yz}}
\label{a-yz}

\paragraph{Lemma \ref{l-yz}.}
\textit{Let $\arule^y$ and $\arule^z$ be two allocation rules such that
the following properties hold for all $\t_i,w_i,\t_j,w_j$:
\beq\label{ea-yz1}
y(\t_i,w_i)=0\quad\Leftrightarrow\quad z(\t_i,w_i)=0\,,
\eeq
\beq
\Bigl\{y(\t_i,w_i)> y(\t_j,w_j)\quad\text{AND}\quad z(\t_i,w_i)<
z(\t_j,w_j)\Bigr\}\quad\Rightarrow\quad\phi(\t_i)w_i>\phi(\t_j)w_j\,. \label{ea-yz2}
\eeq
Then, $R(\arule^y)\geq R(\arule^z)$. Furthermore, if $\arule^y$ and
$\arule^z$ differ with positive probability then $R(\arule^y)> R(\arule^z)$.}

\begin{proof}
Given a realisation $(\bt,\bw)$, suppose there are $k$ advertisers who
receive positive scores. Take the labelling of advertisers:
\begin{align}
&y(\t_1,w_1)> y(\t_2,w_2)>\cdots >y(\t_k,w_k)\nn\\
&\arule_1^y(\bt,\bw)\geq \arule_2^y(\bt,\bw)\geq \cdots\geq \arule_k^y(\bt,\bw)\,.\label{ea-yzx}
\end{align}
Recall that in Sect.~\ref{s-model} we make the assumption of strict
heterogeneity of slot effects ($s_1>s_2>\cdots$). We specify the weak
ordering in \eqref{ea-yzx} because the number of available slots may
be less than $k$. However, if advertiser $i$ receives a positive
allocation then the strict inequality $\arule_i^y>\arule_{i+1}^y$ holds.

From now on we will use the shorthand notation $y(\t_1,w_1)=y_1$,
$\arule_1^y(\bt,\bw)=\arule_1^y$ etc. Let $\G$ be the permutation of indices such
that
\begin{align*}
&z_{\G(1)}> z_{\G(2)}>\cdots> z_{\G(k)}\\
&\arule_{\G(1)}^z\geq \arule_{\G(2)}^z\geq \cdots\geq \arule_{\G(k)}^z\,.
\end{align*}
That is, if an advertiser has the $i$'th highest score w.r.t.\ $z$,
he has the $\G(i)$'th highest score w.r.t.\ $y$.

We need to
show that $\G$ can be reordered through a sequence of swaps, each of
which either satisfies the conditions of Lemma~\ref{l-swap}, or is a
trivial swap. Let $\clS$ be the set of inversions
\[
\clS=\bigl\{\bigl(\G(i),\G(j)\bigr):i<j\,\,\text{and}\,\,\G(i)>\G(j)\bigr\}\,.
\]
We use the fact that $\G$ (and $\G^{-1}$) can be decomposed
into a product of $|\clS|$ adjacent transpositions, where each
transposition resolves precisely one of the inversions in
$\clS$. Applying such a decomposition to the allocations $\arule^z$,
each non-trivial swap resolves some inversion
$\bigl(\G(i),\G(j)\bigr)$. Note that
\begin{itemize}
\item $z_{\G(i)}>z_{\G(j)}$ as $i<j$.
\item $y_{\G(i)}< y_{\G(j)}$ as $\G(i)>\G(j)$.
\item $\phi(\t_{\G(i)})w_{\G(i)}<\phi(\t_{\G(j)})w_{\G(j)}$ as
  \eqref{ea-yz2} holds.
\item $\arule_{\G(i)}> \arule_{\G(j)}$ as the inversion has not been
  previously resolved, and the swap is non-trivial.
\end{itemize}
By the repeated application of Lemma~\ref{l-swap}, given an arbitrary
realisation $(\bt,\bw)$ at which the allocation rules $\arule^y$ and
$\arule^y$ differ,
\[
\sum_{i=1}^n\phi(\t_i)w_i\arule_i^y(\bt,\bw)> \sum_{i=1}^n\phi(\t_i)w_i\arule_i^z(\bt,\bw)\,.
\]
This process can be applied to all realisations, showing that $R(\arule^y)\geq
R(\arule^z)$ pointwise. Furthermore, if $\arule^y$ and
$\arule^z$ differ with positive probability then $R(\arule^y)> R(\arule^z)$.
\qed\end{proof}

\section{Proof of Lemma~\ref{l-Rdom}}
\label{a-Rdom}

\paragraph{Lemma \ref{l-Rdom}.}
\textit{
Given a reserve price $r\in(0,\bart]$, let
\begin{align}
&y(\t_i,w_i)=(\t_i-r)^+w_i\,,\label{ea-Rdomy}\\
&z(\t_i,w_i)=\I\{\t_i\geq r\}\,\t_iw_i\,.\label{ea-Rdomz}
\end{align}
If
\beq\label{ea-Rdomr}
r\leq\inf_{t \geq r}\Bigl\{t-\frac{\phi(t)}{\phi'(t)}\Bigr\}\,,
\eeq
then $R(\arule^y)> R(\arule^z)$.}

\begin{proof}
Consider two advertisers $i$ and $j$ where $\t_i>\t_j\geq r$. Using the
shorthand
notation $\phi(\t_i)=\phi_i$, we write the ratio of their virtual
values as
\[
\frac{\phi_i}{\phi_j}=\frac{\t_i-g(\t_i,\t_j)}{\t_j-g(\t_i,\t_j)}\,,
\]
where
\beq\label{ea-Rdomg}
g(\t_i,\t_j)=\frac{\t_j\phi_i-\t_i\phi_j}{\phi_i-\phi_j}\,.
\eeq
Note that
\[
\frac{\partial}{\partial k}\Bigl(\frac{\t_i-k}{\t_j-k}\Bigr)=\frac{\t_i-\t_j}{(\t_j-k)^2}>0\,.
\]
If $r\leq g(\t_i,\t_j)$,
\[
\frac{\t_i}{\t_j}<\frac{\t_i-r}{\t_j-r}\leq\frac{\phi_i}{\phi_j}\,.
\]
In this case the following properties hold:
\beq\label{e-Rdomprop1}
\biggl\{\frac{\t_i-r}{\t_j-r}>\frac{w_j}{w_i}\quad\text{AND}\quad \frac{\t_i}{\t_j}<\frac{w_j}{w_i}\biggr\}\quad\Rightarrow\quad\frac{\phi_i}{\phi_j}>\frac{w_j}{w_i}\,,
\eeq
\vspace{-2mm}
\beq\label{e-Rdomprop2}
\frac{\t_i-r}{\t_j-r}<\frac{w_j}{w_i}\quad\Rightarrow\quad\frac{\t_i}{\t_j}<\frac{w_j}{w_i}\,.
\eeq
We would like to find under what conditions $r\leq g(\t_i,\t_j)$ for
all $\t_i>\t_j\geq r$.

Denote the upper bound of the range of $\t_i$ by $T$ (possibly
infinite), and consider the infimum of $g(\t_i,\t_j)$. We
argue that this must occur either at one of the limit
points as $\t_j\rightarrow\t_i$ for some $\t_i\in[r,T]$, or at
$(\t_i,\t_j)=(T,r)$.

Consider the
minimising value of $\t_i$ given a fixed $\t_j=t$. This is
  either (a) at $\t_i=T$, (b) at some $\t_i\in(t,T)$, or (c) at the
  limit as $\t_i\rightarrow t$. In case (b), we must have $\partial 
g/\partial\t_i=0$ as $g$ is continuous. From \eqref{ea-Rdomg}
this is equivalent to 
\[
\phi_i'=\frac{\phi_i-\phi(t)}{\t_i-t}\,,
\]
in which case $g$ can be rewritten as
\[
g(\t_i,t)=\t_i-\frac{\phi_i}{\phi_i'}=\lim_{\t_j\rightarrow\t_i}g(\t_i,\t_j)\,.
\]
Thus in both cases (b) and (c), the infimum of $g(\t_i,t)$ w.r.t.\ $t$ is some
limit point of $g(\t_i,\t_j)$ as $\t_j\rightarrow\t_i$.

A similar argument can be made regarding the
minimising value of $\t_j$ given a fixed $\t_i$, leading us to the
conclusion that the infimum value of $g(\t_i,\t_j)$ must occur either
(i) at one of the limit points as $\t_j\rightarrow\t_i$ or (ii) at
$(\t_i,\t_j)=(T,r)$. In case (ii)
\[
r\leq\bart\quad\Rightarrow\quad\phi(r)\leq0\quad\Rightarrow\quad r\leq
\inf \bigl\{g(\t_i,\t_j)\bigr\}\,.
\]
In case (i), we require condition \eqref{ea-Rdomr} to obtain
$r\leq\inf \bigl\{g(\t_i,\t_j)\bigr\}$.

Therefore, if \eqref{ea-Rdomr} holds then $r\leq g(\t_i,\t_j)$ for
all $\t_i>\t_j\geq r$. Then properties \eqref{e-Rdomprop1} and
\eqref{e-Rdomprop2} imply \eqref{ea-yz2} holds and we can invoke
Lemma~\ref{l-yz} to show  
$R(\arule^y)\geq R(\arule^z)$. Furthermore, as $\t_i$ has continuous
support over its range, the allocation rules $\arule^y$ and $\arule^z$
must differ with positive probability, implying $R(\arule^y)> R(\arule^z)$.
\qed\end{proof}

\section{Proof of Lemma~\ref{l-zSNE}}
\label{a-zSNE}

\paragraph{Lemma \ref{l-zSNE}.}
\textit{
Let $\arule$ be an allocation rule that selects a SNE of a
GSP auction with the ranking function $z(b,w)=\I\{b\geq
r\}\,bw$.  Then $R(\arule^z)\geq R(\arule)$.}

\begin{proof}
Given a realisation $(\bt,\bw)$, suppose the allocation rule $\arule$ selects
the SNE in which advertiser $i$ bids $b_i(\bt,\bw)$. We make the
following intuitive assumptions about advertisers' bidding strategies:
\begin{enumerate}
\item $\t_i\geq r\quad\Rightarrow\quad r\leq b_i\leq\t_i$.\label{ass2}
\item $x_i=0\quad\Rightarrow\quad b_i=\t_i$.\label{ass3}
\end{enumerate}
Assumption~\ref{ass2} makes sense as if $\t_i\geq r$, then
$b_i=r$ is a dominant strategy over $b_i<r$ and $b_i=\t_i$ is dominant
over $b_i>\t_i$. Assumption~\ref{ass3} is a little less intuitive, but
is a common concept in auction theory --- that is, any losing bidder
submits the maximum bid without exposing himself to the possibility of
a loss, which is clearly a (weakly) dominant strategy and further drives
competition in the auction. Assumptions~\ref{ass2} and
\ref{ass3} imply
\beq\label{ea-zSNEprop1}
z(\t_i,w_i)=0\quad\Leftrightarrow\quad z(b_i,w_i)=0\,.
\eeq
Suppose there are $k$ qualifying
advertisers ($\t_i,b_i\geq r$). Take the labelling of advertisers such that
\begin{align*}
&b_1w_1>b_2w_2>\cdots>b_kw_k\\
&x_1\geq x_2\geq \cdots\geq x_k\,.
\end{align*}
Consider any realisation at which the allocation rules $\arule$ and
$\arule^z$ differ. That is, there exists a pair of advertisers $j<i$
such that $b_jw_j>b_iw_i$, $\t_jw_j<\t_iw_i$, and $x_j>x_i$.
From the SNE inequalities \eqref{ea-zSNEineq} we have 
\begin{align}\label{ea-zSNEineq2}
\bigl(\t_jw_j-\max\{rw_j,b_{j+1}w_{j+1}\}\bigr)x_j\geq
\bigl(\t_jw_j-\max\{rw_j,b_{i+1}w_{i+1}\}\bigr)x_i\\
\label{ea-zSNEineq3}
\bigl(\t_iw_i-\max\{rw_i,b_{j+1}w_{j+1}\}\bigr)x_j\leq
\bigl(\t_iw_i-\max\{rw_i,b_{i+1}w_{i+1}\}\bigr)x_i\,.
\end{align}
Taking \eqref{ea-zSNEineq3} away from \eqref{ea-zSNEineq2}:
\begin{align*}
(\t_jw_j-\t_iw_i)(x_j-x_i)\geq
&\max\bigl\{rw_i,b_{i+1}w_{i+1}\bigr\}-\max\bigl\{rw_i,b_{j+1}w_{j+1}\bigr\}\\
&+\max\bigl\{rw_j,b_{j+1}w_{j+1}\bigr\}-\max\bigl\{rw_j,b_{i+1}w_{i+1}\bigr\}\,.
\end{align*}
As $x_j>x_i$ and $\t_jw_j<\t_iw_i$, the LHS (and thus the RHS
also) is negative. In the interest of brevity, denote the four terms in the RHS by
$A$, $B$, $C$, and $D$ respectively. We know $C\geq D$ (as
$b_{j+1}w_{j+1}> b_{i+1}w_{i+1}$), thus it is necessary that
$A<B$. This implies $B=b_{j+1}w_{j+1}$ as $A\geq rw_i$. This implies
$C\geq B$, and thus it is necessary that $A<D$. This implies $D=rw_j$
as $A\geq b_{i+1}w_{i+1}$. Now we have
\begin{align*}
\text{RHS}&=\max\bigl\{rw_i,b_{i+1}w_{i+1}\bigr\}-b_{j+1}w_{j+1}+\max\bigl\{rw_j,b_{j+1}w_{j+1}\bigr\}-rw_j\\
&=\max\bigl\{rw_i,b_{i+1}w_{i+1}\bigr\}-\min\bigl\{rw_j,b_{j+1}w_{j+1}\bigr\}\,.
\end{align*}
For this to be negative, it is necessary that $rw_j>rw_i$
and thus $w_j>w_i$. As $\t_jw_j<\t_iw_i$, we need $\t_j<\t_i$. As the
hazard rate $f(\t_i)/(1-F(\t_i))$ is non-decreasing,
\begin{align*}
\frac{1-F(\t_i)}{f(\t_i)}&\leq \frac{1-F(\t_j)}{f(\t_j)}\\
\frac{1-F(\t_i)}{f(\t_i)}\,w_i&\leq \frac{1-F(\t_j)}{f(\t_j)}\,w_j\,.
\end{align*}
Taking this inequality away from $\t_iw_i>\t_jw_j$, we get
$\phi(\t_i)w_i>\phi(\t_j)w_j$. Thus,
\beq\label{ea-zSNEprop2}
\Bigl\{\t_iw_i> \t_jw_j\quad\text{AND}\quad
  b_iw_i< b_jw_j\Bigr\}\quad\Rightarrow\quad\phi(\t_i)w_i>\phi(\t_j)w_j\,.
\eeq
Note that \eqref{ea-zSNEprop2} coupled with \eqref{ea-zSNEprop1} closely
resemble the properties \eqref{ea-yz1} and \eqref{ea-yz2} required for
Lemma~\ref{l-yz}. Indeed, one can follow the same process described in
Appendix~\ref{a-yz} and show that $R(\arule^z)\geq R(\arule)$.
\qed\end{proof}

\end{document}